\documentclass[aps,reprint,superscriptaddress,nofootinbib,showkeys]{revtex4-2}

\usepackage[T1]{fontenc}
\usepackage[utf8]{inputenc}
\usepackage{lmodern}
\usepackage{amsmath,amssymb,mathtools}
\usepackage{amsthm}
\usepackage{graphicx}
\usepackage{booktabs,array}
\usepackage{dcolumn}
\usepackage{bm}
\usepackage{xcolor}
\usepackage{listings}
\usepackage{microtype}
\usepackage{url}
\usepackage[colorlinks=true,linkcolor=blue!50!black,citecolor=blue!50!black,urlcolor=blue!50!black]{hyperref}

\newtheorem{proposition}{Proposition}
\newtheorem{remark}{Remark}

\newcommand{\R}{\mathbb{R}}
\newcommand{\Cset}{\mathcal{C}}
\newcommand{\Hhist}{\mathcal{H}}
\newcommand{\Fset}{\mathcal{F}}
\newcommand{\med}{\operatorname{median}}
\newcommand{\MAD}{\operatorname{MAD}}
\newcommand{\RankScale}{\operatorname{RankScale}}
\newcommand{\clip}{\operatorname{clip}}
\newcommand{\cosw}{\operatorname{cos}_{w}}
\newcommand{\LoL}{\textit{League of Legends}}

\hypersetup{%
    pdftitle={Robust Player-Conditional Champion Ranking for League of Legends},
    pdfauthor={Min Heo, Pranav Kadiyam and Prasun Panthi}
}

\begin{document}

\title{Robust Player-Conditional Champion Ranking for \LoL{}:\\
Style Similarity, Mastery Priors, and Archetype-Constrained Discovery}

\author{Min Heo}
\affiliation{Wabash College, Crawfordsville, IN 47933, USA}

\author{Pranav Kadiyam}
\affiliation{Arizona State University, Tempe, AZ 85281, USA}

\author{Prasun Panthi}
\affiliation{Wabash College, Crawfordsville, IN 47933, USA}

\date{May 3, 2026}

\begin{abstract}
Champion recommendation in multiplayer online battle arena games is usually framed informally as a problem of metagame strength, personal comfort, or global win rate. We formalize champion recommendation in \LoL{} as an interpretable, player-conditional ranking problem under sparse, noisy, and non-stationary behavioral data. The proposed framework combines four information sources: a population-strength proxy, player-style similarity, direct and indirect mastery priors, and archetype-level guardrails. The method uses robust median/MAD normalization, logarithmic transforms for skewed event counts, recency-weighted player style vectors, mastery-weighted champion-pool vectors, weighted cosine similarity, rank-scaled score components, and k-means++ clustering for coarse archetype support. The implemented prototype uses a Python/Pandas modeling layer, Supabase-backed storage, and a web-facing recommendation interface. Unlike black-box supervised win-prediction systems, the proposed method returns decomposed recommendation scores that can be inspected as expected-performance proxy, fit, mastery, and archetype compatibility. A single-player case study on a 100-game history for the player identifier \texttt{DIVINERAINRACCON} is included as an end-to-end sanity check. The manuscript is therefore a methods and systems contribution: it specifies a reproducible, modular, and auditable champion recommender and gives a validation protocol for future large-scale evaluation through temporal train-test splits, next-champion recovery, calibration analysis, and ablation studies.
\end{abstract}


\maketitle

\section{Introduction}

Champion selection in \LoL{} is a recommender-systems problem disguised as a game decision. A player does not simply ask which champion is strong in the population. The more precise question is: which champion is strong enough, stylistically compatible enough, and familiar enough to be worth recommending to this specific player at this point in their history?

A naive recommender might rank champions by global win rate, popularity, or patch strength. This is insufficient for at least three reasons. First, match outcomes are highly confounded: teammates, opponents, draft order, lane assignments, patch changes, role swaps, side selection, and player execution all influence the result. Second, individual histories are sparse. A player may have hundreds of games overall but only a small number on any single champion. Third, champion suitability is not one-dimensional. A player with high farm, low death rate, and scaling-carry habits may not benefit from the same recommendations as a player whose strongest signals are roaming, vision, durability, or crowd control.

This paper introduces an interpretable multi-signal recommendation framework for champion selection. The method treats a champion as an item, the player as a user, and gameplay statistics as behavioral traces. However, instead of using opaque supervised learning, it builds a transparent scoring pipeline. Each candidate champion receives separate scores for population strength, style fit, mastery/familiarity, and archetype compatibility. The final recommendation is therefore auditable: the system can say whether a champion is recommended because it is statistically strong, because it resembles the player's recent behavior, because the player already has mastery on it, or because it is an adjacent discovery pick.

\subsection{Contributions}

The main contributions are as follows.

\begin{enumerate}
    \item We formulate champion recommendation as a content-based, player-conditional ranking problem with decomposable utility components.
    \item We define a robust preprocessing pipeline using median/MAD normalization, sign correction for negative metrics, clipping, and log transforms for skewed event counts.
    \item We construct player representations from both recency-weighted match behavior and mastery-weighted champion-pool behavior.
    \item We introduce a weighted cosine fit score whose feature weights depend on both population dispersion and player-specific behavioral salience.
    \item We define direct mastery, direct performance, and indirect familiarity scores, allowing the recommender to distinguish comfort picks from statistically adjacent discovery picks.
    \item We add an archetype guardrail based on k-means++ clustering over broader champion feature vectors, reducing the risk of recommending superficially similar but strategically distant champions.
    \item We describe a working implementation with a Python modeling stack, Supabase data layer, and web-facing interface, and I give a single-player end-to-end case study.
\end{enumerate}

\subsection{Scope of claims}

This is a methods and systems preprint. It does not claim that the recommender improves win rate in live play, nor does it claim superiority over baselines without large-scale validation. The case study in Section~\ref{sec:case-study} is an execution trace and qualitative sanity check, not a statistically powered experiment. The empirical contribution required for a stronger recommender-systems paper would be a temporal held-out evaluation over many players and patches, with baselines and ablations. Section~\ref{sec:validation} specifies such an evaluation protocol.

\section{Related Work}

Recommender systems combine item information, user behavior, and ranking objectives to help users choose among many possible items \cite{resnick1997recommender, ricci2015handbook, jannach2010recommender}. The proposed model is closest to content-based and hybrid recommendation, where item features and user profiles are explicitly represented \cite{lops2011content, burke2002hybrid}. Unlike collaborative-filtering systems based on many user-item ratings \cite{koren2009matrix, hu2008implicit}, champion selection lacks clean explicit ratings. Win/loss labels are also noisy implicit feedback because a single match is affected by many variables outside the target player's champion choice.

Learning-to-rank and implicit-feedback methods are natural future directions \cite{joachims2002optimizing, liu2009learning, rendle2009bpr}. In the present prototype, however, we deliberately avoid training a black-box ranker. The goal is interpretability and auditability rather than immediate leaderboard optimization. This design choice follows a common theme in explainable recommendation: a useful system should not only rank items but also provide reasons that are faithful to the scoring mechanism \cite{zhang2020explainable, ribeiro2016should}.

The statistical preprocessing is motivated by robust statistics. Game metrics are often heavy-tailed, sparse, and role-dependent. Median absolute deviation (MAD) is a standard robust scale estimator and is less sensitive to outliers than the sample standard deviation \cite{rousseeuw1993mad, hampel1986robust}. The archetype module uses k-means++ style seeding, which improves the reliability of k-means initialization \cite{arthur2007kmeanspp}; general clustering concerns and interpretation limits follow the broader clustering literature \cite{jain2010data}.

The engineering layer is built around Riot-related data access and public analytics workflows. Riot's developer ecosystem and Data Dragon resources make it possible to build applications around \LoL{} data \cite{riotdeveloper}. Supabase provides a Postgres-backed storage and Python client interface for tabular retrieval \cite{supabasepython}. OP.GG is used in this paper only as qualitative, player-facing context for screenshots of match history and mastery, not as a formal source of ground-truth model labels \cite{opgg}.

\section{Problem Formulation}

Let \(\Cset\) be a set of candidate champions and let \(\Hhist = (h_1,\ldots,h_T)\) be the observed history of a player. Each champion \(c\in\Cset\) is represented by a population feature vector
\[
    x_c=(x_{c1},\ldots,x_{cd})\in\R^d,
\]
and each player-history row \(h_i\) is represented by a feature vector
\[
    y_i=(y_{i1},\ldots,y_{id})\in\R^d.
\]
The recommender returns a ranked list
\[
    c_{(1)}, c_{(2)}, \ldots, c_{(K)}
\]
where \(c_{(1)}\) has the highest recommendation score.

The target score is not a calibrated probability of winning. Instead, we define a decomposed utility
\[
    R(c) = \Phi\big(W(c),F(c),M(c),G(c),T(c)\big),
\]
where \(W(c)\) is an expected-performance proxy, \(F(c)\) is style fit, \(M(c)\) is mastery/familiarity, \(G(c)\) is archetype compatibility, and \(T(c)\) is a support term. This distinction matters. In the current overall-mode implementation, the population table may not contain a reliable role-specific win-rate column. Therefore \(W(c)\) should be read as a performance proxy rather than as \(\Pr(\mathrm{win}\mid c,\mathrm{player})\).

\section{Feature Space and Data Model}

The implementation uses two overlapping feature sets. The primary recommendation feature set is shown in Table~\ref{tab:features}. The values are prototype weights, not universal constants.

\begin{table*}[t]
\caption{Overall recommendation features and prototype weights.}
\label{tab:features}
\begin{tabular*}{\textwidth}{@{\extracolsep{\fill}}lll}
\hline\hline
Feature & Weight & Interpretation \\
\hline
\texttt{damagePerMinute} & 0.20 & Damage output intensity.\\
\texttt{goldPerMinute} & 0.16 & Resource acquisition.\\
\texttt{cs\_per\_min} & 0.14 & Farming profile.\\
\texttt{laneMinionsFirst10Minutes} & 0.10 & Early-lane farming.\\
\texttt{deaths\_per\_min} & 0.18 & Survival discipline; sign reversed.\\
\texttt{killParticipation} & 0.10 & Skirmish and teamfight involvement.\\
\texttt{damageDealtToBuildings} & 0.06 & Structure pressure.\\
\texttt{damageDealtToObjectives} & 0.03 & Neutral-objective pressure.\\
\texttt{visionScorePerMinute} & 0.02 & Vision contribution.\\
\texttt{totalTimeCCDealt} & 0.01 & Crowd-control contribution.\\
\hline\hline
\end{tabular*}
\end{table*}

For archetype modeling, the feature set is expanded to include durability and utility variables. In implementation notation, the archetype feature set is
\[
\begin{aligned}
\Fset_{\mathrm{arch}}=\Fset_S\cup\{&\texttt{totalDamageTaken},\\
&\texttt{damageSelfMitigated}\}.
\end{aligned}
\]
Equivalently, it consists of damage, economy, farming, early-lane farming, death control, kill participation, durability, mitigation, vision, structure pressure, objective pressure, and crowd-control features. The broader set prevents archetype labels from collapsing into damage and economy alone.

\section{Robust Preprocessing}

\subsection{Skewed counts and sign conventions}

Sparse events such as objective steals, baron takedowns, dragon takedowns, rift herald takedowns, turret plates, and turret takedowns are log-transformed:
\[
    x \mapsto \log(1+x), \qquad x\geq 0.
\]
This preserves ordering while limiting the leverage of rare event outliers. Directionally negative metrics are sign-corrected. For example, after normalization, \(\texttt{deaths\_per\_min}\) is multiplied by \(-1\), so fewer deaths contribute positively.

\subsection{Median/MAD normalization}

For a feature vector \(v=(v_1,\ldots,v_n)\), define
\[
    m=\med(v), \qquad \MAD(v)=\med_i |v_i-m|.
\]
The robust score is
\[
    z_i = 0.67448975 \frac{v_i-m}{\MAD(v)}.
\]
If the MAD is zero or undefined, the transformed score is set to zero. Scores are then clipped:
\[
    z_i \leftarrow \clip(z_i,-3,3).
\]
The constant \(0.67448975\) makes the MAD estimator comparable to standard deviation under a Gaussian reference model. The purpose is not to assume normality, but to use a familiar scale while maintaining robustness.

\begin{remark}
Ordinary mean-standard-deviation scaling is fragile in this setting because game metrics can be sparse, heavy-tailed, role-dependent, and patch-sensitive. Robust scaling is a modeling choice.
\end{remark}

\section{Population Strength}

Let \(w_j^{(S)}\) be the feature weight for strength and let \(\widetilde z_{cj}\) be the sign-adjusted robust z-score for champion \(c\). The raw population-strength proxy is
\[
    S_{\mathrm{raw}}(c)=\sum_{j\in \Fset_S} w_j^{(S)} \widetilde z_{cj}.
\]
This raw score is converted to a percentile-like rank score:
\[
    S(c)=\RankScale(S_{\mathrm{raw}}(c)).
\]
If all values are tied or unavailable, the implementation assigns a neutral value of \(0.5\).

A role-aware prototype also exists. When reliable role labels are available, feature weights can differ for top, jungle, middle, bottom, and utility. For champion \(c\) in role \(r\), with \(n_{cr}\) games, average row score \(\bar s_{cr}\), observed win rate \(\hat p_{cr}\), role baseline \(\bar p_r\), shrinkage constant \(K\), and row-score standard deviation \(\sigma_{cr}\), one can define
\[
    \mathrm{skill}_{cr}=\frac{n_{cr}}{n_{cr}+K}\bar s_{cr},
\]
\[
    \mathrm{edge}_{cr}=\frac{n_{cr}}{n_{cr}+K}(\hat p_{cr}-\bar p_r),
\]
and
\[
    \mathrm{penalty}_{cr}=\lambda\frac{\sigma_{cr}}{\sqrt{n_{cr}}}.
\]
The role-aware score is
\[
    \mathrm{score}_{cr}=\mathrm{skill}_{cr}+\beta\,\mathrm{edge}_{cr}-\mathrm{penalty}_{cr}.
\]
This version is statistically preferable when role labels and win-rate columns are reliable. The overall-mode recommender described below is more conservative when those columns are incomplete.

\section{Player Style Representation}

Let the preprocessed player history contain \(T\) rows. The model constructs two player vectors: a recent-game vector and a champion-pool vector.

\subsection{Recent-game vector}

For decay parameter \(\rho>0\), define recency weights
\[
    a_i=\frac{\exp(-\rho(T-i))}{\sum_{k=1}^T \exp(-\rho(T-k))}.
\]
The recent-game vector is
\[
    u_{\mathrm{game}}=\sum_{i=1}^T a_i y_i.
\]
In the prototype, \(\rho=0.18\). Recent games matter more, but older games do not vanish entirely.

\subsection{Champion-pool vector}

Let \(\Cset_{\mathrm{played}}\subseteq \Cset\) be the player's observed champions. For each played champion \(c\), let \(\bar y_c\) be the mean player-history vector on that champion, \(P_c\) be champion mastery points, \(L_c\) be champion level, and \(G_c\) be observed game count. A direct mastery raw score is
\[
    M_{d,\mathrm{raw}}(c)
    =0.60\log(1+P_c)+0.20L_c+0.20\log(1+G_c).
\]
After rank scaling and lower clipping of the weights, the champion-pool vector is
\[
    u_{\mathrm{pool}}=\frac{\sum_{c\in\Cset_{\mathrm{played}}} \omega_c\bar y_c}
    {\sum_{c\in\Cset_{\mathrm{played}}}\omega_c}.
\]
The recent-game vector captures current form; the pool vector captures longer-run identity.

\section{Style Fit by Weighted Cosine Similarity}

For vectors \(a,b\in\R^d\) and positive feature weights \(w=(w_1,\ldots,w_d)\), define
\[
    \cosw(a,b)=
    \frac{\sum_{j=1}^d w_j a_j b_j}
    {\sqrt{\sum_{j=1}^d w_j a_j^2}\sqrt{\sum_{j=1}^d w_j b_j^2}}.
\]
Let \(\sigma_j\) be the population standard deviation of normalized champion feature \(j\). Let \(u_{\mathrm{game},j}\) be the player's recent-game coordinate. The unnormalized feature attention is
\[
    q_j=\sigma_j(1+\alpha |u_{\mathrm{game},j}|),
\]
and
\[
    w_j=\frac{q_j}{\sum_{k=1}^d q_k}.
\]
The hyperparameter \(\alpha\) controls how strongly the model emphasizes features where the player deviates from average. In the prototype, \(\alpha=0.75\).

For candidate champion \(c\), raw fit is
\[
    F_{\mathrm{raw}}(c)
    =0.55\cosw(u_{\mathrm{game}},x_c)+0.45\cosw(u_{\mathrm{pool}},x_c).
\]
The final fit score is
\[
    F(c)=\RankScale(F_{\mathrm{raw}}(c)).
\]

\begin{proposition}[Scale invariance of weighted cosine]
If \(\lambda,\mu>0\), then \(\cosw(\lambda a,\mu b)=\cosw(a,b)\), provided the denominators are nonzero.
\end{proposition}

\begin{proof}
The numerator is multiplied by \(\lambda\mu\). Each norm factor is multiplied by \(\lambda\) and \(\mu\), respectively. The multiplicative factors cancel.
\end{proof}

\begin{remark}
The fit model compares behavioral direction, not absolute magnitude. This is appropriate after robust normalization because the question is whether the player's statistical identity resembles the champion's statistical identity.
\end{remark}

\section{Mastery, Performance, and Indirect Familiarity}

\subsection{Direct signals}

For each champion \(c\) in the player's history, define
\[
\begin{aligned}
G_c &= \text{observed player games on } c,\\
P_c &= \text{mastery points},\\
L_c &= \text{mastery level}.
\end{aligned}
\]
Let \(\bar R_c\) be the mean player row score and \(\bar W_c\) be the mean observed win indicator. Let \(A_c\) be the total recency mass assigned to games on \(c\). The direct mastery raw score is
\[
\begin{aligned}
M_{d,\mathrm{raw}}(c)=&\ 0.55\log(1+P_c)+0.15L_c\\
&+0.20\log(1+G_c)+0.10A_c.
\end{aligned}
\]
The direct performance raw score is
\[
    P_{d,\mathrm{raw}}(c)=0.65\bar R_c+0.35(\bar W_c-0.5).
\]
Both are rank-scaled:
\[
\begin{aligned}
M_d(c)&=\RankScale(M_{d,\mathrm{raw}}(c)),\\
P_d(c)&=\RankScale(P_{d,\mathrm{raw}}(c)).
\end{aligned}
\]

\subsection{Indirect familiarity}

Let \(\mathcal M\) be the set of champions with positive direct mastery. For candidate \(c\) and mastered champion \(m\in\mathcal M\), define
\[
    s(c,m)=\max\{0,\cos(x_c,x_m)\}.
\]
Let \(\eta_m\) be the rank-scaled direct mastery weight of \(m\). The indirect familiarity raw score is the mean of the top \(K\) weighted similarities:
\[
    M_{i,\mathrm{raw}}(c)=
    \operatorname{mean}\left(\operatorname{TopK}_{m\in\mathcal M}\{\eta_m s(c,m)\}\right).
\]
After rank scaling,
\[
    M_i(c)=\RankScale(M_{i,\mathrm{raw}}(c)).
\]
The combined mastery/familiarity score is
\[
    M(c)=0.70M_d(c)+0.30M_i(c).
\]
This allows the recommender to make two different kinds of recommendations: comfort recommendations for champions the player has already used, and discovery recommendations for champions near the player's mastered pool.

\section{Archetype Guardrail}

Feature similarity can be misleading. Two champions may both have high damage and low deaths while differing in strategic role, teamfight function, map responsibility, and learning cost. The archetype guardrail reduces this risk.

Let \(X_{\mathrm{arch}}\) be the matrix of champion vectors over \(\Fset_{\mathrm{arch}}\). The implementation clusters champions using k-means with k-means++ initialization and multiple restarts, keeping the assignment with lowest inertia:
\[
    \sum_{c\in\Cset}\|x_c-\mu_{a(c)}\|_2^2,
\]
where \(a(c)\) is the cluster assignment.

Cluster centroids are labeled using deterministic style dimensions: damage, farm, tankiness, utility, siege pressure, and teamfight involvement. The labels are interpretation aids rather than ontological claims. Candidate labels include frontline tank, artillery control, scaling carry, siege splitpush, utility support, and skirmish bruiser.

Let \(B(c)\) be the rank-scaled support score for the archetype containing \(c\). The archetype guardrail is
\[
    G(c)=0.55B(c)+0.45\max\{M_d(c),M_i(c)\}.
\]
This permits a champion to be protected either by archetype support or by direct/indirect familiarity.

\section{Final Scoring Function}

Let \(G_c\) be the number of player games observed on champion \(c\). Define player-confidence weight
\[
    \gamma(c)=\frac{G_c}{G_c+3}.
\]
For low-data champions, define a fallback utility
\[
    U(c)=0.55S(c)+0.45F(c).
\]
The expected-performance proxy is
\[
    W(c)=\gamma(c)P_d(c)+(1-\gamma(c))U(c).
\]
The base score is
\[
    Q(c)=\lambda_W W(c)+\lambda_F F(c)+\lambda_M M(c),
\]
with default weights
\[
    (\lambda_W,\lambda_F,\lambda_M)=(0.50,0.25,0.25).
\]
The support score is
\[
    T(c)=0.60F(c)+0.40\max\{M_d(c),M_i(c)\}.
\]
The support multiplier is
\[
    H(c)=0.82+0.18T(c).
\]
The archetype multiplier is
\[
    A(c)=
    \begin{cases}
    0.90+0.10G(c), & G_c>0,\\
    0.72+0.28G(c), & G_c=0.
    \end{cases}
\]
The final recommendation score is
\[
    R(c)=Q(c)H(c)A(c).
\]

\begin{proposition}[Boundedness]
Assume \(S,F,M,M_d,M_i,P_d,G,T\in[0,1]\) and \(\lambda_W,\lambda_F,\lambda_M\geq0\) with \(\lambda_W+\lambda_F+\lambda_M=1\). Then \(R(c)\in[0,1]\).
\end{proposition}

\begin{proof}
Since \(U(c)\), \(P_d(c)\), and \(\gamma(c)\) lie in \([0,1]\), \(W(c)\in[0,1]\). Hence \(Q(c)\) is a convex combination of values in \([0,1]\), so \(Q(c)\in[0,1]\). Also \(H(c)=0.82+0.18T(c)\in[0.82,1]\). If \(G_c>0\), then \(A(c)\in[0.90,1]\); if \(G_c=0\), then \(A(c)\in[0.72,1]\). Therefore \(R(c)=Q(c)H(c)A(c)\in[0,1]\).
\end{proof}

\begin{remark}
Boundedness is used to ensure score comparability and numerical stability. 
\end{remark}

\section{Algorithmic Summary}

\begin{enumerate}
    \item Load player-history and population tables.
    \item Validate required columns such as \texttt{championName}, \texttt{championPoints}, and \texttt{championLevel}.
    \item Coerce expected features to numeric values.
    \item Apply \(\log(1+x)\) to skewed event counts.
    \item Compute robust z-scores using population distributions.
    \item Reverse sign for negative metrics.
    \item Compute population strength and rank-scale it.
    \item Construct \(u_{\mathrm{game}}\) and \(u_{\mathrm{pool}}\).
    \item Compute weighted cosine fit for each candidate champion.
    \item Aggregate direct mastery and direct performance by champion.
    \item Compute indirect familiarity from similarity to mastered champions.
    \item Cluster candidates into archetypes and compute archetype support.
    \item Compute \(W(c)\), \(F(c)\), \(M(c)\), \(G(c)\), \(T(c)\), and \(R(c)\).
    \item Return the top-\(K\) recommendations with decomposed explanation fields.
\end{enumerate}

\section{Implementation Architecture}

The prototype is implemented in Python using NumPy and Pandas. The key implementation modules are:

\begin{itemize}
    \item \texttt{league\_stats.py}: role-aware champion scoring, robust preprocessing, shrinkage, stability penalties, and champion aggregation.
    \item \texttt{player\_recommender.py}: overall-mode recommender with player vectors, weighted cosine fit, direct/indirect mastery, archetypes, and final scoring.
    \item \texttt{supabase\_only\_recommender.py}: Supabase-native loading and recommendation execution.
    \item \texttt{test\_supabase.py}: environment loading and Supabase connectivity checks.
\end{itemize}

The intended backend interface is an HTTP endpoint such as:
\begin{lstlisting}
POST /recommend
{
  "gameName": "PLAYER_NAME",
  "tagLine": "NA1",
  "topN": 30
}
\end{lstlisting}
The frontend displays ranked champion cards with component scores, qualitative explanations, and filtering controls. Fig.~\ref{fig:frontend} shows the prototype interface.

\begin{figure*}[t]
    \centering
    \includegraphics[width=0.95\textwidth]{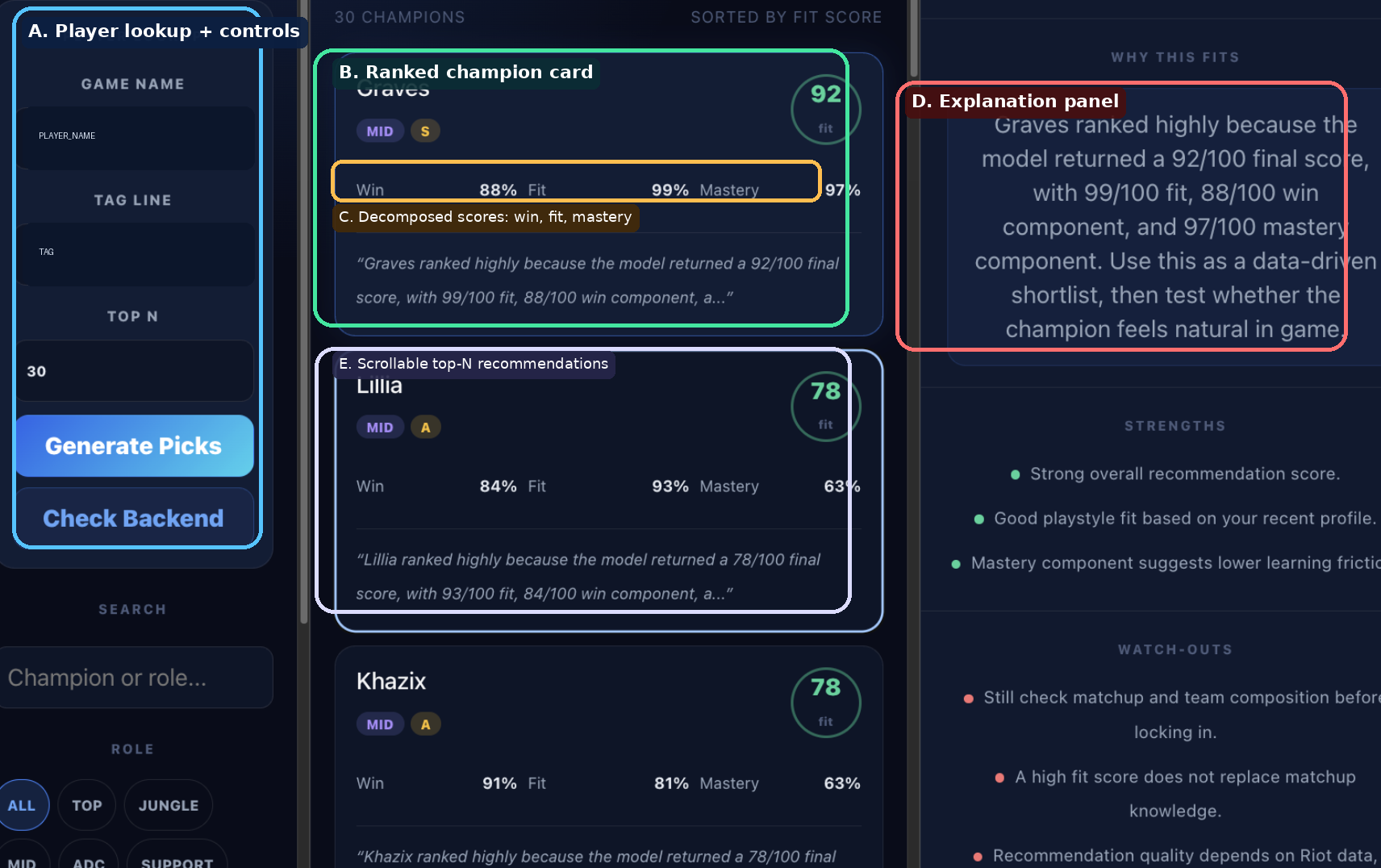}
    \caption{Prototype frontend interface. The interface exposes player lookup controls, a scrollable top-\(N\) recommendation list, decomposed win/fit/mastery components, and a natural-language explanation panel. This screenshot is a UI smoke test rather than a formal evaluation result.}
    \label{fig:frontend}
\end{figure*}

\section{Case Study: End-to-End Run on \texttt{DIVINERAINRACCON}}
\label{sec:case-study}

This section reports a single-player execution trace. It is included to show that the system can run end-to-end on a nontrivial player history and to illustrate how the decomposed scores should be interpreted. It is not a statistically powered experiment.

The recommender was run on an anonymized development dataset associated with the player identifier \texttt{DIVINERAINRACCON}. The run contained \(100\) player games and \(51\) unique champions. The role mix in the recorded history was Middle: 42, Top: 19, Utility: 15, Bottom: 14, and Jungle: 10.
The default weights were
\[
    (\lambda_W,\lambda_F,\lambda_M)=(0.50,0.25,0.25),
\]
with \(\alpha=0.75\). The three highest-supported archetypes were scaling carry, frontline tank, and utility support. Because this run used the overall-mode population table, the \texttt{win\_score} column is an expected-performance proxy rather than a calibrated win probability.

\begin{table*}[t]
\centering
\caption{Top ten recommendations for the \texttt{DIVINERAINRACCON} case study. Scores are rounded to three decimals.}
\label{tab:case-study}
\resizebox{\textwidth}{!}{
\begin{tabular}{llrrrrrr}
\toprule
Champion & Type & Final & Win proxy & Fit & Mastery & Guardrail & Similarity \\
\midrule
Malphite & comfort\_or\_known & 0.804 & 0.741 & 0.994 & 0.874 & 0.620 & 0.749\\
Heimerdinger & comfort\_or\_known & 0.791 & 0.851 & 0.842 & 0.722 & 0.947 & 0.437\\
Xerath & comfort\_or\_known & 0.785 & 0.844 & 0.737 & 0.808 & 0.997 & 0.366\\
Ornn & comfort\_or\_known & 0.751 & 0.721 & 1.000 & 0.742 & 0.554 & 0.826\\
Anivia & comfort\_or\_known & 0.750 & 0.699 & 0.936 & 0.882 & 0.426 & 0.587\\
Mordekaiser & comfort\_or\_known & 0.744 & 0.644 & 0.871 & 0.899 & 0.890 & 0.509\\
Annie & comfort\_or\_known & 0.742 & 0.712 & 0.901 & 0.882 & 0.411 & 0.535\\
Hwei & comfort\_or\_known & 0.736 & 0.676 & 0.889 & 0.900 & 0.521 & 0.525\\
Nasus & comfort\_or\_known & 0.727 & 0.686 & 0.947 & 0.768 & 0.570 & 0.613\\
Cho'Gath & comfort\_or\_known & 0.723 & 0.667 & 0.895 & 0.843 & 0.589 & 0.534\\
\bottomrule
\end{tabular}}
\end{table*}

The top recommendation, Malphite, demonstrates why a decomposed score is preferable to a single opaque output. Malphite is not recommended because of population strength. Its case-study fit score is \(0.994\), mastery score is \(0.874\), and final score is \(0.804\). The model therefore interprets Malphite as a high-fit comfort recommendation, not just as a generic metagame pick. Ornn provides a complementary example: it achieves a perfect fit score of \(1.000\), but its final score remains below Malphite's because the final rank also depends on expected-performance proxy, mastery, support, and archetype effects.

The mage and control recommendations are also coherent. Heimerdinger, Xerath, Anivia, Annie, and Hwei appear in the top ten, matching a history with many middle-lane games and multiple control/damage champions. The presence of frontline picks such as Malphite, Ornn, Nasus, Cho'Gath, and Mordekaiser  follows from the full 100-game history and the model's archetype support. The OP.GG screenshots in Figures~\ref{fig:opgg-history} and~\ref{fig:opgg-mastery} should be read as qualitative context showing recent match patterns and mastery information, not as the formal source of the scoring table.

\begin{figure*}[t]
    \centering
    \includegraphics[height=0.76\textheight]{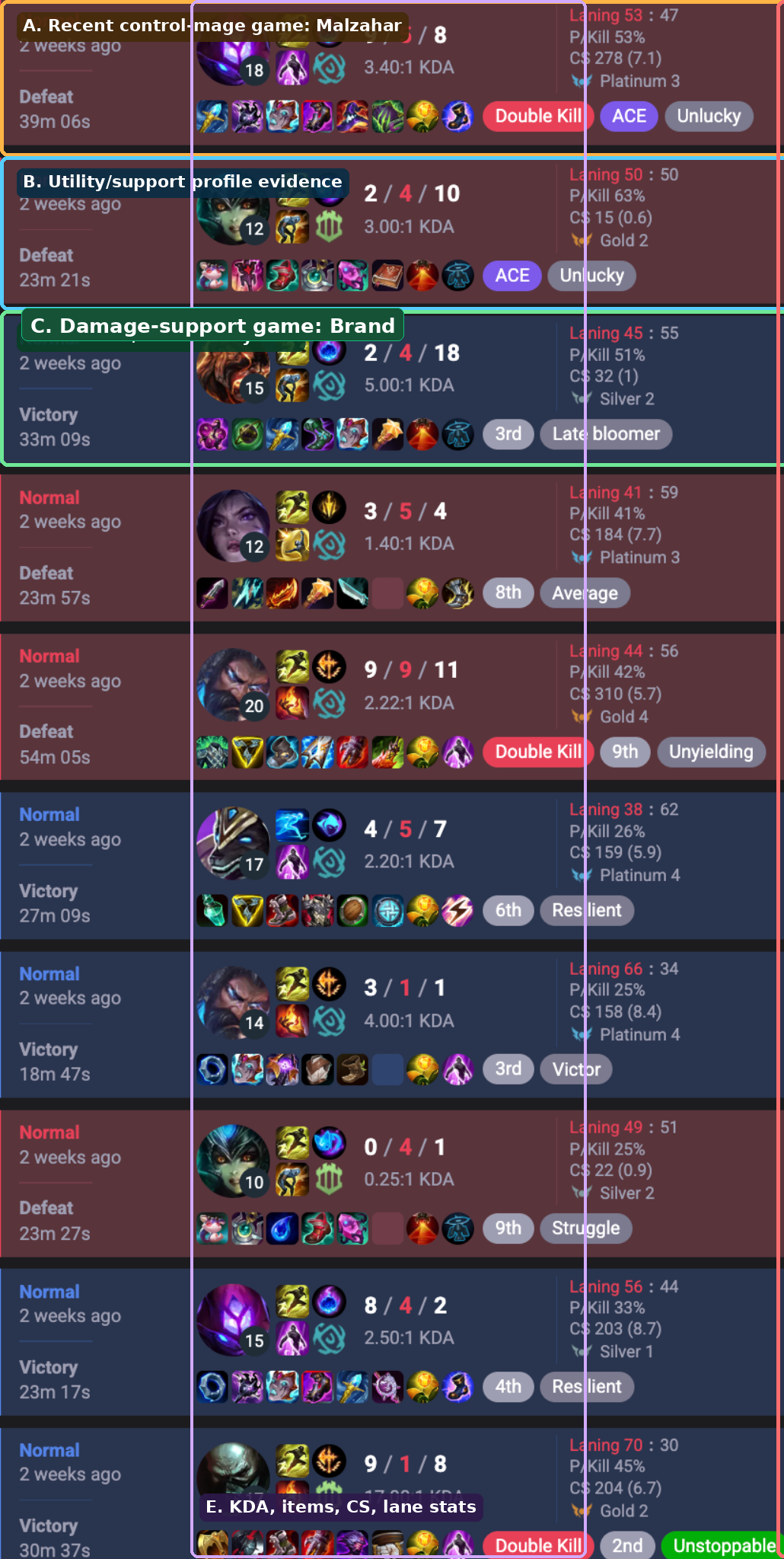}
    \caption{Qualitative OP.GG match-history context for the case study. The image has been cropped to reduce exposure of other player names. The annotations highlight recent mage/control play, utility-support evidence, and a damage-support Brand row. This figure is not used as a formal numerical input; it contextualizes why mage/control and support-adjacent recommendations are plausible. Source: OP.GG player-facing match-history interface \cite{opgg}.}
    \label{fig:opgg-history}
\end{figure*}

\begin{figure*}[t]
    \centering
    \includegraphics[width=0.55\textwidth]{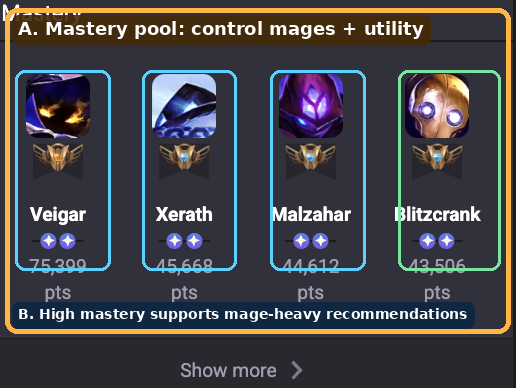}
    \caption{Qualitative OP.GG champion mastery context. Veigar and Xerath mastery help explain why control-mage recommendations such as Xerath, Anivia, Annie, Hwei, and Heimerdinger are plausible. Source: OP.GG player-facing mastery interface \cite{opgg}.}
    \label{fig:opgg-mastery}
\end{figure*}

\subsection{Score audit for Malphite}

For Malphite in Table~\ref{tab:case-study}, the reported components are approximately
\[
    W=0.741,\qquad F=0.994,\qquad M=0.874.
\]
The base score is therefore
\[
    Q = 0.50(0.741)+0.25(0.994)+0.25(0.874)\approx 0.8375.
\]
The final score is lower than \(Q\) because the support and archetype multipliers are conservative. This is desirable: the system does not simply select the largest fit score. It requires a candidate to be supported by multiple dimensions. The numerical audit also makes clear that the recommendation is interpretable: changing the mastery weight, the archetype penalty, or the fit-sensitivity parameter would change the ranking in inspectable ways.

\section{Validation Protocol}
\label{sec:validation}

A public empirical version of this work should evaluate the recommender using temporal held-out data. The following protocol is intended to prevent look-ahead bias and to distinguish preference recovery from performance prediction.

\subsection{Temporal next-champion recovery}

For each player, sort matches by time. At time \(t\), train or compute the recommender using only matches \(1,\ldots,t\), then rank all candidate champions. If the player actually chooses \(c_{t+1}\), compute:
\[
    \mathrm{Hit@K} = \frac{1}{N}\sum_{t=1}^N \mathbf{1}\{c_{t+1}\in L_t^{(K)}\},
\]
and
\[
    \mathrm{MRR} = \frac{1}{N}\sum_{t=1}^N \frac{1}{\operatorname{rank}_t(c_{t+1})}.
\]
This evaluates whether the recommender captures player tendency and champion-pool structure. It does not by itself prove that the recommendations improve win rate.

\subsection{Outcome calibration}

For each historical match, record the model's score \(R(c_t)\) for the champion actually selected. A simple calibration model is
\[
    \Pr(\mathrm{win}_t=1)=\sigma(a+bR(c_t)),
\]
where \(\sigma\) is the logistic function. A positive \(b\) would suggest that the score contains outcome-related information, but causal interpretation would still require controls for draft, matchup, role, side, teammates, opponents, and patch.

\subsection{Baselines}

The minimum baseline set should include:
\begin{enumerate}
    \item most-played champion;
    \item highest-mastery champion;
    \item most-recently played champion;
    \item population-strength-only ranking;
    \item ordinary cosine similarity without mastery or archetype guardrail;
    \item random champion within role or within the candidate pool.
\end{enumerate}

\subsection{Ablations}

The following ablations should be reported:
\begin{enumerate}
    \item remove population strength;
    \item remove fit;
    \item remove direct mastery;
    \item remove indirect familiarity;
    \item remove archetype guardrail;
    \item replace robust z-scores with ordinary z-scores;
    \item replace recency weighting with uniform averaging;
    \item replace weighted cosine with ordinary cosine.
\end{enumerate}
A serious predictive claim requires the full model to improve ranking quality or calibration relative to these alternatives.

\section{Data, Privacy, and Reproducibility}

The prototype uses two table types: a player-history table and a population champion table. The implementation supports Supabase loading with environment variables for URL, key, player table, population table, and optional player filters. Missing credentials and empty tables raise explicit errors.

For a public release, raw match-history data should be anonymized or aggregated. Screenshots that include other players' names should be cropped, redacted, or replaced by synthetic diagrams. The figures in this paper are included as qualitative development documentation; any public repository should include a data card specifying source, collection date, patch range, queue type, filtering criteria, and privacy treatment.

A reproducible release should include:
\begin{enumerate}
    \item the recommendation source code;
    \item a small anonymized sample dataset;
    \item instructions for constructing the population table;
    \item a script that reproduces Table~\ref{tab:case-study};
    \item an evaluation script for Hit@K, MRR, calibration, baselines, and ablations.
\end{enumerate}

\section{Limitations}

\subsection{The win score is not a true win probability}

In the overall-mode prototype, the population table may not contain role-specific win-rate data. Therefore the \texttt{win\_score} should be interpreted as an expected-performance proxy, not a calibrated probability of winning.

\subsection{Patch dependence}

Champion strength and playstyle distributions change across patches. A public evaluation should either restrict to a patch window or include patch-aware time weighting.

\subsection{Role ambiguity}

Role-aware weighting is statistically preferable when role labels are reliable. The overall-mode recommender is robust to incomplete data, but it loses role specificity.

\subsection{Outcome confounding}

Win/loss outcomes are influenced by teammates, enemies, draft order, side, role, matchup, player state, and execution. Direct win aggregation is therefore noisy and should not be interpreted causally.

\subsection{Human preference versus optimization}

Players may want different objectives: climbing, comfort, learning, role coverage, meta conformity, or enjoyment. The weights \((\lambda_W,\lambda_F,\lambda_M)\) should eventually be user-controllable.

\section{Future Work}

A natural extension is Bayesian champion strength estimation. If champion \(c\) has \(w_c\) wins in \(n_c\) games, a beta-binomial model gives
\[
    p_c\sim \operatorname{Beta}(\alpha_0,\beta_0),\qquad
    w_c\mid p_c\sim \operatorname{Binomial}(n_c,p_c),
\]
with posterior mean
\[
    \mathbb{E}[p_c\mid w_c,n_c]=
    \frac{\alpha_0+w_c}{\alpha_0+\beta_0+n_c}.
\]
This would provide principled shrinkage for population win strength.

Another extension is draft expected value. Let \(D\) be a partial draft state including ally picks, enemy picks, bans, side, and role assignment. Then the target becomes
\[
    \arg\max_{c\in\Cset(D)} \mathbb{E}[\mathrm{utility}\mid c,D,\mathrm{player}].
\]
This is harder because it requires modeling champion interactions, matchup effects, role constraints, and team composition.

Finally, once larger datasets are collected, the transparent score can be compared against learning-to-rank methods or used as a feature in a supervised ranker. The goal should not be to abandon interpretability, but to test whether empirical optimization improves ranking while preserving faithful explanations.

\section{Conclusion}

This paper presents an interpretable multi-signal champion recommender for \LoL{}. The key idea is that recommendation quality should not be reduced to global win rate or player mastery alone. A useful recommendation arises from the interaction of population strength, player style, direct experience, indirect familiarity, and archetype compatibility.

The proposed method is mathematically explicit and implementation-oriented. It uses robust normalization, recency-weighted player vectors, weighted cosine similarity, direct and indirect mastery, archetype clustering, and conservative multipliers. It also produces decomposed outputs suitable for explanation in a frontend interface.

The current manuscript should be read as a formal methods preprint and prototype report. It includes a single-player execution trace, but it does not claim validated population-level predictive superiority. The next stage is a temporal held-out evaluation with baselines, ablations, calibration analysis, and patch-aware robustness checks.

\appendix

\section{Backend Output Fields}

The backend returns a list of candidate champions together with a metadata block. The most important recommendation fields are \texttt{championName}, \texttt{recommendation\_type}, \texttt{archetype\_name}, \texttt{final\_score}, \texttt{win\_score}, \texttt{fit\_score}, \texttt{mastery\_score}, \texttt{archetype\_guardrail}, \texttt{population\_strength\_score}, \texttt{direct\_mastery\_score}, \texttt{indirect\_mastery\_score}, \texttt{player\_games}, and \texttt{similarity\_raw}. The public-facing interface should expose only a subset of these values; the remaining fields are primarily diagnostic.

\section{Role-Aware Feature Emphasis}

The role-aware prototype uses different feature priorities across positions. Top-lane scoring emphasizes damage, farming, death control, structure pressure, and durability. Jungle scoring emphasizes kill participation, objective control, vision, and death control. Middle-lane scoring emphasizes damage, gold, farming, lane pressure, and kill participation. Bottom-lane scoring emphasizes damage, gold, farming, lane minions, and death control. Utility scoring emphasizes vision, crowd control, kill participation, and objective support. The overall-mode recommender used in the case study is a fallback mixture over damage, economy, farming, vision, objectives, and structure pressure.

\section{Notation Index}

\begin{center}
\begin{tabular}{ll}
\hline\hline
Symbol & Meaning\\
\hline
$\Cset$ & Candidate champions.\\
$\Hhist$ & Player match history.\\
$x_c$ & Population vector for champion $c$.\\
$y_i$ & Player-history vector for row $i$.\\
$S(c)$ & Population-strength proxy.\\
$F(c)$ & Style-fit score.\\
$M(c)$ & Mastery/familiarity score.\\
$W(c)$ & Expected-performance proxy.\\
$G(c)$ & Archetype guardrail score.\\
$R(c)$ & Final recommendation score.\\
\hline\hline
\end{tabular}
\end{center}

\end{document}